\documentclass[floatfix, aip, jcp, reprint, twocolumn, amsmath, amssymb, graphicx]{revtex4-2}

\usepackage{dcolumn} 
\usepackage{bm} 
\usepackage{amsthm} 
\usepackage{amsmath}
\usepackage{amssymb}
\usepackage{bbold}

\usepackage{amsfonts}
\usepackage{braket}
\usepackage{enumerate}
\usepackage{xcolor}
\usepackage{ulem}
\usepackage{simplewick}
\usepackage{graphicx}
\graphicspath{{fig/}}

\newcommand{\T}{\mathrm{T}}
\newcommand{\pf}{\ensuremath{\mathrm{pf}}}
\newcommand{\hpf}{\ensuremath{\mathrm{hpf}}}

\draft 

\begin{document}


\title
{Robust Formulation of Wick's theorem for Computing Matrix Elements Between
Hartree--Fock--Bogoliubov Wavefunctions}



\author{Guo P. Chen}
\email[]{guo.chen@rice.edu}
\affiliation{Department of Chemistry, Rice University, Houston, Texas 77005-1892, USA}

\author{Gustavo E. Scuseria}
\affiliation{Department of Chemistry, Rice University, Houston, Texas 77005-1892, USA}
\affiliation{Department of Physics and Astronomy, Rice University, Houston, Texas 77005-1892, USA\looseness=-1}


\date{April 26, 2023. Revised: May 25, 2023}

\begin{abstract}

Numerical difficulties associated with computing matrix elements of operators
between Hartree--Fock--Bogoliubov (HFB) wavefunctions have plagued the development of
HFB-based many-body theories for decades. 
The problem arises from divisions by zero in the standard formulation of the
nonorthogonal Wick's theorem in the limit of vanishing HFB overlap.
In this paper, we present a robust formulation of Wick's theorem that
stays well-behaved regardless of whether the HFB states are orthogonal or not.
This new formulation ensures cancellation between the zeros of the overlap and
the poles of the Pfaffian, which appears naturally in fermionic systems.
Our formula explicitly eliminates self-interaction, which otherwise causes
additional numerical challenges. A computationally efficient version of our
formalism enables robust symmetry-projected HFB calculations with the same
computational cost as mean-field theories.
Moreover, we avoid potentially diverging normalization factors
by introducing a robust normalization procedure.
The resulting formalism treats even and odd number of
particles on equal footing and reduces to Hartree--Fock as a natural limit.
As proof of concept, we present a numerically stable and accurate solution to a
Jordan--Wigner-transformed Hamiltonian, whose singularities motivated the
present work. Our robust formulation of Wick's theorem is a most promising
development for methods using quasiparticle vacuum states.


\end{abstract}

\maketitle 

\section{Introduction}
  \label{sec:intro}

The Hartree--Fock--Bogoliubov (HFB) wavefunction is a fermionic mean-field
ansatz allowing particle number symmetry breaking.\cite{Ring1980Nuclear, Blaizot1986Quantum} It
encompasses the Hartree--Fock (HF) wavefunction in the same sense that
spin-unrestricted HF (UHF) reduces to spin-restricted HF (RHF) without spin
$S^2$ symmetry breaking. HFB-based methods are not as conventional in quantum
chemistry as in nuclear structure theory. This is a consequence of repulsive
electronic interactions in the former, which results in number symmetry not
breaking spontaneously in mean-field.\cite{Lieb1981Variational} 

Nevertheless, the emergence of variation-after-projection approaches to
symmetry breaking and restoration \cite{Sheikh2000PHFB, Schmid2004PHFB,
Scuseria2011PQT, Sheikh2021Restoration} has sparked growing interest in HFB and
related anz{\"a}tze in electronic structure theory. \cite{Scuseria2011PQT,
Naftchi2011Rank, Neuscamman2012Size, Neuscamman2013Jastrow, Carlos2012N,
Uemura2015AGPCI, Takashi2015DMET, Degroote2016POST, Shi2017Stocastic,
Mahajan2019Symmetry, Khamoshi2019AGPRDM,
Henderson2020CorrelatingAGP, Dutta2020GeminalReplacement, Larsson2020MMPS, Johnson2020RG,
Khamoshi2021AGPQC, Lin2021HFB, Khamoshi2021JCC, Dutta2021LCAGP, Khamoshi2023AGPCCQC}
Methods based on symmetry-projected HFB (PHFB) are particularly promising for
describing strongly correlated systems, where traditional HF-based methods
fail.\cite{Scuseria2011PQT, Yang1962ODLRO, Coleman1965AGP, Sager2022CPC}

Essential for the implementation of PHFB-based methods are matrix elements of
many-body operators between nonorthogonal HFB states. Indeed, the PHFB
formalism can be considered a general form of nonorthogonal configuration
interaction (NOCI).\cite{HillWheeler1953, Peierls1957Collective}  The
computation of relevant matrix elements relies on the generalized Wick's
theorem with respect to nonorthogonal HFB states, also known as the
nonorthogonal Wick's theorem.\cite{Onishi1966Generator, Balian1969Wick} 

However, when the overlap between right and left HFB states vanishes, the
standard formula in the nonorthogonal Wick's theorem becomes numerically
ill-conditioned, and eventually ill-defined, causing large round-off errors.
This problem occurs in PHFB calculations\cite{Tajima1992Generator, Anguiano2001Particle} 
and is exacerbated by self-interaction\cite{Perdew1981SIC, Bender2009Particle}
if exchange- or pairing-type contractions are neglected, \textit{e.g.},
within the context of nuclear density functional theory.
\cite{Dobaczewski2007NDFT, Lacroix2009Configuration, Bender2009Particle, Duguet2009Particle}

Our recent work mapping spin systems to fermions using the Jordan--Wigner
transformation\cite{Jordan1928JW, Henderson2022Duality} is another example
where orthogonal HFB wavefunctions arise and lead to large numerical errors. This happens
when fermionic on-site occupations are equal to $1/2$, which is frequently
encountered in frustrated phases of transformed spin systems.

The main objective of this paper is to demonstrate that with proper algorithmic
design, the nonorthogonal Wick's theorem can be extended to the orthogonal
limit, resulting in accurate and efficient computation of matrix elements
between arbitrary HFB states. Related algorithms have been discussed for
HF-based NOCI,\cite{Koch1993LinearSuperposition, Thom2009NOCI, RodriguezLaguna2020Efficient,
Burton2021Wick} 
and our formalism reduces to them in that limit.
Remedies for the numerical problems encountered in the HFB case have been
proposed for specific applications\cite{Tajima1992Generator,
Anguiano2001Particle} 
or under certain restrictions.
\cite{Doenau1998Canonical, Dobaczewski2000GBMZ} However, a universal,
low-scaling formula for robust computation of matrix elements between HFB
states, as the one presented here, was lacking.

\section{Nonorthogonal Wick's Theorem}

We define unnormalized HFB wavefunctions or \textit{quasiparticle vacuum states},
$\ket{\Phi_0}$ and $\ket{\Phi_1}$, as
\begin{equation}
  \label{eq:hfb}
  \ket{\Phi_j} = \prod_{p=1}^M \beta^j_p \ket{-}
\end{equation}
for $j = 0, 1$, where $\ket{-}$ denotes the physical vacuum,
and $M$ is the dimension of the one-particle (spin-orbital) basis.
The Bogoliubov quasiparticle annihilation operators
$\beta^j_p$ are related to the
fermionic particle operators $c^\dag_q$ and $c_q$
through a unitary canonical transformation known as the \textit{Bogoliubov transformation}:
\begin{equation}
  \beta^j_p
  = \sum_{q=1}^M \left(U^{j*}_{qp} c_q + V^{j*}_{qp} c^\dag_q\right).
\end{equation}
To ensure nonvanishing $\ket{\Phi_j}$, we assume that $V^j$ is non-singular,
\textit{i.e.}, all of the canonical orbitals
of $\ket{\Phi_j}$ obtained from the Bloch--Messiah decomposition\cite{Bloch1962Canonical}
of $U^j$ and $V^j$ (\textit{vide infra}) are at least infinitesimally occupied.
We will eliminate this assumption and address the normalization of $\ket{\Phi_j}$
in Sec.~\ref{sec:normalization}.

The nonorthogonal Wick's theorem can be concisely stated
using the \textit{Pfaffian} as follows:\cite{Hu2014MatrixElements}
For $\braket{\Phi_0|\Phi_1} \neq 0$,
a matrix element of a $d$-body operator is expressed as
\begin{equation}
  \label{eq:nonorthogonal_wick}
  \braket{\Phi_0|\gamma^1 \gamma^2 \cdots \gamma^{2d}|\Phi_1}
  = \braket{\Phi_0|\Phi_1} \pf(\Gamma),
\end{equation}
where $\{\gamma^k\}_{1\leq k \leq 2d}$ is a set of arbitrary fermionic operators, and
$\Gamma$ is a $2d \times 2d$ antisymmetric matrix whose strict upper triangular part
is defined by \textit{contractions} of the form
\begin{equation}
  \label{eq:contraction}
  \Gamma_{kl}
  = \contraction[1.5ex]{}{\gamma}{{}^k}{\gamma}
  \gamma^k \gamma^l
  = \frac{\braket{\Phi_0|\gamma^k \gamma^l|\Phi_1}}{\braket{\Phi_0|\Phi_1}}
\end{equation}
for $k < l$. The Pfaffian $\pf(\Gamma)$ automatically generates
all possible full contractions of the fermionic operators with the correct signs,
reflecting the fermionic anticommutation relations.
For example, Eq.~\eqref{eq:nonorthogonal_wick} reduces to the conventional
statement of the nonorthogonal Wick's theorem
\cite{Onishi1966Generator, Balian1969Wick, Ring1980Nuclear, Blaizot1986Quantum}
for the two-particle reduced transition matrix ($2$-RTM),
\cite{Lowdin1955QTMBS1, McWeeny1960DMT, Mazziotti1999Comparison}
or second-order transition density matrix:
\label{eq:2rtm}
\begin{align}
  \nonumber
  \mathcal{D}^{(2)}_{pq,p'q'}
  &= \braket{\Phi_0|c^\dag_p c^\dag_q c_{q'} c_{p'}|\Phi_1}\\
  \nonumber
  &= \braket{\Phi_0|\Phi_1}
  \big(
    \contraction[2ex]{}{c}{{}^\dag_p c^\dag_q c_{q'}}{c}
    \contraction[1.2ex]{c^\dag_p}{c}{{}^\dag_q}{c}
    c^\dag_p c^\dag_q c_{q'} c_{p'}
    + \contraction[2ex]{}{c}{{}^\dag_p c^\dag_q}{c}
    \contraction[1.2ex]{c^\dag_p}{c}{{}^\dag_q c_{q'}}{c}
    c^\dag_p c^\dag_q c_{q'} c_{p'}
    + \contraction[1.2ex]{}{c}{{}^\dag_p}{c}
    \contraction[1.2ex]{c^\dag_p c^\dag_q}{c}{{}_{q'}}{c}
    c^\dag_p c^\dag_q c_{q'} c_{p'}
  \big)\\
  &= \braket{\Phi_0|\Phi_1}
  \big(
    \rho^{01}_{p'p} \rho^{01}_{q'q}
    - \rho^{01}_{p'q} \rho^{01}_{q'p}
    + \kappa^{10*}_{pq} \kappa^{01}_{p'q'}
  \big),
\end{align}
where
\begin{equation}
  \rho^{01}_{pq}
  = \contraction[1.2ex]{}{c}{{}^\dag_q}{c}
  c^\dag_q c_p, \quad
  \kappa^{01}_{pq}
  = \contraction[1.2ex]{}{c}{{}_q}{c}
  c_q c_p, \quad
  \kappa^{10*}_{pq}
  = \contraction[1.2ex]{}{c}{{}^\dag_p}{c}
  c^\dag_p c^\dag_q,
\end{equation}
and we follow the convention that each crossing between contraction lines
introduces a minus sign, which is implied by the properties of the Pfaffian.
In general, we can write $\gamma^k$ as a quasiparticle
\begin{equation}
  \label{eq:gammadef}
  \gamma^k
  = \sum_{q=1}^M \left(A^{*}_{qk} c_q + B^{*}_{qk} c^\dag_q\right),
\end{equation}
where $A$ and $B$ are $M \times 2d$ matrices.
Note that for $k \neq l$, $\gamma^k$ and $\gamma^l$ may come from different
Bogoliubov transformations, hence the superscript in our notation.

We see from Eq.~\eqref{eq:contraction} that the nonorthogonal Wick's theorem becomes
ill-defined when $\braket{\Phi_0|\Phi_1} = 0$. In this case, we may 
fall back to the original Wick's theorem with respect to the physical vacuum.
\cite{Wick1950Wick, Gaudin1960Wick, Lieb1968Pfaffian}
We have\cite{Bertsch2012Pfaffian}
\begin{widetext}
\begin{subequations}
\label{eq:wick_gaudin}
\begin{align}
  \braket{\Phi_0|\gamma^1 \gamma^2 \cdots \gamma^{2d}|\Phi_1}
  &= (-1)^{M(M-1)/2} \braket{-|
    \beta^{0\dag}_{1} \beta^{0\dag}_{2} \cdots \beta^{0\dag}_{M}
    \gamma^1 \gamma^2 \cdots \gamma^{2d}
    \beta^1_1 \beta^1_2 \cdots \beta^1_{M}
  |-}\\
  &= (-1)^{M(M-1)/2}
  \pf\begin{pmatrix}
    V^{0\,\T} U^0  & V^{0\,\T} B^*  & V^{0\,\T} V^{1*}\\
    -B^\dag V^0    & \Gamma^{(-)}   & A^\dag V^{1*}\\
    -V^{1\dag} V^0 & -V^{1\dag} A^* & U^{1\dag} V^{1*}
  \end{pmatrix},
\end{align}
\end{subequations}
\end{widetext}
where $\Gamma^{(-)}$ is a $2d \times 2d$ antisymmetric matrix defined by
\begin{equation}
  \Gamma^{(-)}_{kl}
  = \braket{-|\gamma^k \gamma^l|-}
\end{equation}
for $k < l$.
Compared with Eq.~\eqref{eq:nonorthogonal_wick}, Eq.~\eqref{eq:wick_gaudin}
is always well-defined, but the Pfaffian of a much larger matrix needs to be evaluated,
greatly increasing the computational cost.

In the following, we sketch a proof of the nonorthogonal Wick's theorem
stated in Eq.~\eqref{eq:nonorthogonal_wick}, 
introducing our notation along the way.
Unlike the proof by induction presented in Ref.~\onlinecite{Hu2014MatrixElements},
our proof follows a more direct approach inspired by Ref.~\onlinecite{Mizusaki2012Formulation}.
We start by rewriting Eq.~\eqref{eq:wick_gaudin}
using the properties of the Pfaffian:
\begin{equation}
  \label{eq:pf_m}
  \braket{\Phi_0|\gamma^1 \gamma^2 \cdots \gamma^{2d}|\Phi_1}
  = (-1)^{M(M-1)/2} \pf(\mathcal{M}),
\end{equation}
where
\begin{align}
  \mathcal{M} 
  = \begin{pmatrix}
    \mathcal{S}     & \mathcal{G}\\
    -\mathcal{G}^\T & \Gamma^{(-)}
  \end{pmatrix},
\end{align}
\begin{align}
  \label{eq:def_s_g}
  \mathcal{S}
  = \begin{pmatrix}
    V^{0\,\T} U^0  & V^{0\,\T} V^{1*}\\
    -V^{1\dag} V^0 & U^{1\dag} V^{1*}
  \end{pmatrix}, \quad
  \mathcal{G}
  = \mathcal{V}^\T \mathcal{C}^*
\end{align}
and
\begin{equation}
  \label{eq:def_v_c}
  \mathcal{V}
  = \begin{pmatrix}
    0   & -V^{1*}\\
    V^0 & 0
  \end{pmatrix},
  \quad
  \mathcal{C}
  = \begin{pmatrix}
    A\\ B
  \end{pmatrix}.
\end{equation}
When $\braket{\Phi_0|\Phi_1} \neq 0$, the $2M \times 2M$ antisymmetric matrix $\mathcal{S}$
is non-singular since\cite{Bertsch2012Pfaffian}
\begin{equation}
  \label{eq:overlap}
  \braket{\Phi_0|\Phi_1}
  = (-1)^{M(M-1)/2} \pf(\mathcal{S}).
\end{equation}
We therefore have the following Pfaffian identity
\begin{align}
  \label{eq:pf_block_m}
  \pf(\mathcal{M})
  = \pf(\mathcal{S}) \pf(\mathcal{M} / \mathcal{S}),
\end{align}
where $\mathcal{M}/\mathcal{S}$ denotes the Schur complement of the block $\mathcal{S}$
of the supermatrix $\mathcal{M}$, \textit{i.e.},
\begin{align}
  \label{eq:superrdm}
  \mathcal{M} / \mathcal{S}
  = \Gamma^{(-)} + \mathcal{G}^\T \mathcal{S}^{-1} \mathcal{G}.
\end{align}
On the other hand, as shown in detail in the supplementary material,
we find
\begin{equation}
  \label{eq:defk}
  \mathcal{K}
  = \mathcal{V} \mathcal{S}^{-1} \mathcal{V}^\T
  = \begin{pmatrix}
    -\kappa^{01} & -\rho^{01}\\
    \rho^{01\,\T} & \kappa^{10*}
  \end{pmatrix},
\end{equation}
and it follows that
\begin{equation}
  \label{eq:gamma_schur}
  \mathcal{M}/\mathcal{S}
  = \Gamma^{(-)} + \mathcal{C}^\dag \mathcal{K} \mathcal{C}^*
  = \Gamma.
\end{equation}
Inserting Eqs.~\eqref{eq:overlap} and \eqref{eq:gamma_schur}
into Eqs.~\eqref{eq:pf_m} and \eqref{eq:pf_block_m} completes the proof.

Without loss of generality, we hereafter assume that
$\gamma^1 \gamma^2 \cdots \gamma^{2d}$ are 
particle operators in normal order with respect to the physical vacuum.
Thus, we have $\Gamma^{(-)} = 0_{2d \times 2d}$ and
\begin{equation}
  \label{eq:normal_ordered}
  \braket{\Phi_0|\gamma^1 \gamma^2 \cdots \gamma^{2d}|\Phi_1}
  = (-1)^{M(M-1)/2} \pf(\mathcal{S}) \pf(\mathcal{G}^\T \mathcal{S}^{-1} \mathcal{G}).
\end{equation}
These particular matrix elements contain those that arise in the $d$-particle or $d$th-order
reduced transition matrix \mbox{($d$-RTM)},
\cite{Lowdin1955QTMBS1, McWeeny1960DMT, Mazziotti1999Comparison}
\begin{equation}
  \mathcal{D}^{(d)}_{p_1 p_2 \cdots p_d, q_1 q_2 \cdots q_d}
  = \braket{\Phi_0|c^\dag_{p_1} c^\dag_{p_2} \cdots c^\dag_{p_d} c_{q_d} \cdots c_{q_2} c_{q_1}|\Phi_1}.
\end{equation}

Eq.~\eqref{eq:normal_ordered} is reminiscent of L{\"o}wdin's formula for $d$-RTM between HF wavefunctions,
\cite{Lowdin1955QTMBS2}
for which singular value decomposition (SVD) can be used to extract the zeros and poles
in order to evaluate the zero-overlap limit.
\cite{Koch1993LinearSuperposition, Thom2009NOCI, RodriguezLaguna2020Efficient, Burton2021Wick}
We show in the next section that this idea can be generalized to the HFB case.

\section{Robust Wick's Theorem}
\label{sec:robust_wick}

The antisymmetric matrix $\mathcal{S}$ can be written in canonical form,
\cite{Hua1944Theory, Wimmer2012Algorithm923}
\begin{equation}
  \label{eq:canon}
  \mathcal{S} =
  \mathcal{Q} \bar{\mathcal{S}} \mathcal{Q}^\T,
\end{equation}
where $\mathcal{Q}$ is unitary and
\begin{align}
  \label{eq:canon_vals}
  \bar{\mathcal{S}}
  = \begin{pmatrix}
    0 & \text{diag}(s)\\
    -\text{diag}(s) & 0
  \end{pmatrix}.
\end{align}
The elements of the vector $s$ are nonnegative
and they reduce to the singular values of the overlap matrix in the HF case,
as shown in the supplementary material. 
Moreover, it follows from Eq.~\eqref{eq:overlap} that
\begin{equation}
  \label{eq:overlap_in_s}
  \braket{\Phi_0|\Phi_1}
  = \zeta \prod_r s_r,
\end{equation}
where $\zeta = \det(\mathcal{Q})$ is a complex phase factor.
We observe that the poles of $\mathcal{S}^{-1}$ in Eq.~\eqref{eq:normal_ordered}
get cancelled out by the zeros of the overlap,
hinting at a more general and robust formula for computing matrix elements.

Define
\begin{equation}
  \tilde{\mathcal{K}}^r
  = \begin{pmatrix}
    -\tilde{\kappa}^{01,r} & -\tilde{\rho}^{01,r}\\
    \left(\tilde{\rho}^{01,r}\right)^\T & \left(\tilde{\kappa}^{10,r}\right)^*
  \end{pmatrix}
\end{equation}
where
\begin{subequations}
\label{eq:new_contractions}
\begin{align}
  \tilde{\rho}^{01,r}_{pq}
  &= \left[L_{11}\right]_{pr} \left[L_{22}\right]_{qr}
  - \left[L_{12}\right]_{pr} \left[L_{21}\right]_{qr},\\
  \tilde{\kappa}^{01,r}_{pq}
  &= \left[L_{11}\right]_{pr} \left[L_{12}\right]_{qr}
  - \left[L_{12}\right]_{pr} \left[L_{11}\right]_{qr},\\
  \left(\tilde{\kappa}^{10,r}_{pq}\right)^*
  &= \left[L_{22}\right]_{pr} \left[L_{21}\right]_{qr}
  - \left[L_{21}\right]_{pr} \left[L_{22}\right]_{qr},
\end{align}
\end{subequations}
and
\begin{equation}
  \label{eq:defl}
  \mathcal{L}
  = \begin{pmatrix}
    L_{11} & L_{12}\\
    L_{21} & L_{22}
  \end{pmatrix}
  = \mathcal{V} \mathcal{Q}^*
\end{equation}
It is straightforward to show
\begin{equation}
  \Gamma = \sum_r s_r^{-1} \tilde{\Gamma}^r, \quad
  \tilde{\Gamma}^r
  = \mathcal{C}^\dag \tilde{\mathcal{K}}^r \mathcal{C}^*.
  \label{eq:def_gammat_r}
\end{equation}

We now propose a robust formulation of Wick's theorem for computing matrix elements between HFB wavefunctions:
\begin{align}
  \nonumber
  &\braket{\Phi_0|\gamma^1 \gamma^2 \cdots \gamma^{2d}|\Phi_1}
  = (2d-1)!!\, \zeta\\
  &\quad \times \sum_{r_1 r_2 \cdots r_d} \lambda_{r_1 r_2 \cdots r_d} \,
  \hpf(\tilde{\Gamma}^{r_1} \otimes \tilde{\Gamma}^{r_2} \otimes \cdots \otimes \tilde{\Gamma}^{r_d})
  \label{eq:robust_wick}
\end{align}
with
\begin{equation}
  \label{eq:deflambda}
  \lambda_{r_1 \cdots r_d} =
  \begin{cases}
    \displaystyle{\prod_{r \neq r_1, \cdots, r_d} s_r} &
    (r_1, \cdots, r_d \;\text{are distinct})\\
    \quad\quad\; 0 & (\text{otherwise})
  \end{cases}
\end{equation}
\begin{align}
  \nonumber
  &\hpf(\tilde{\Gamma}^{r_1} \otimes \tilde{\Gamma}^{r_2} \otimes \cdots \otimes \tilde{\Gamma}^{r_d})
  = \\
  &\quad\: \frac{1}{(2d)!} \sum_{\sigma \in S_{2d}} \text{sgn}(\sigma)
  \tilde{\Gamma}^{r_1}_{\sigma(1) \sigma(2)}
  \tilde{\Gamma}^{r_2}_{\sigma(3) \sigma(4)}
  \cdots
  \tilde{\Gamma}^{r_d}_{\sigma(2d-1) \sigma(2d)}
  \label{eq:generalized_pfaffian}
\end{align}
where $\otimes$ denotes tensor product, $\sigma$ enumerates permutations in the symmetric group $S_{2d}$,
and $\text{sgn}(\sigma)$ is the signature of $\sigma$.
Eq.~\eqref{eq:generalized_pfaffian} defines a special case of the
\textit{hyper-Pfaffian},\cite{Barvinok1995Hyperpfaffian, Luque2002Pfaffian} denoted by $\hpf(\cdot)$,
which generalizes the Pfaffian to tensors.
For the general definition of the hyper-Pfaffian,
we refer the reader to the supplementary material.

The proof of the theorem is presented in the supplementary material.
We emphasize that the formula in Eq.~\eqref{eq:robust_wick} remains
well-behaved with vanishing $\braket{\Phi_0|\Phi_1}$ because
all potential divisions by zero have been prevented by construction.
Specifically, all $s_r^{-1}$ factors are cancelled out by the factors of $\braket{\Phi_0|\Phi_1}$
in Eq.~\eqref{eq:overlap_in_s}, while all terms involving $s_r^{-2}$, $s_r^{-3}$,
\textit{etc.}, drop out due to the cancellation of self-interaction.
The latter is guaranteed by the properties of the Pfaffian and is further enforced numerically
by setting $\lambda_{r_1 \cdots r_d}$ to zero when its indices are not
distinct. From the definition of $\lambda_{r_1 \cdots r_d}$,
we can also see that the matrix element vanishes if the number of zero elements
in $s$ exceeds $d$, consistent with the Slater--Condon rules between HF wavefunctions;
\cite{Slater1929Spectra, Condon1930Spectra, Lowdin1955QTMBS1}
see the next section for generalization to the HFB case.

In practical calculations, we do not need to explicitly evaluate the hyper-Pfaffian for
each individual matrix element.
Instead, Eq.~\eqref{eq:robust_wick} generates expressions
of high-order RTMs
in terms of the modified contractions defined in Eq.~\eqref{eq:new_contractions}.
As an example, the robust Wick's theorem implies the following expression of the $2$-RTM:
\begin{align}
  \nonumber
  \mathcal{D}^{(2)}_{pq,p'q'}
  &= \braket{\Phi_0|
    \contraction[2ex]{}{c}{{}^\dag_p c^\dag_q c_{q'}}{c}
    \contraction[1.2ex]{c^\dag_p}{c}{{}^\dag_q}{c}
    c^\dag_p c^\dag_q c_{q'} c_{p'}
    + \contraction[2ex]{}{c}{{}^\dag_p c^\dag_q}{c}
    \contraction[1.2ex]{c^\dag_p}{c}{{}^\dag_q c_{q'}}{c}
    c^\dag_p c^\dag_q c_{q'} c_{p'}
    + \contraction[1.2ex]{}{c}{{}^\dag_p}{c}
    \contraction[1.2ex]{c^\dag_p c^\dag_q}{c}{{}_{q'}}{c}
    c^\dag_p c^\dag_q c_{q'} c_{p'}
  |\Phi_1}\\
  \nonumber
  &= \zeta \sum_{r_1 r_2} \lambda_{r_1 r_2}
  \Bigl(
    \tilde{\rho}^{01,r_1}_{p'p} \tilde{\rho}^{01,r_2}_{q'q}
    - \tilde{\rho}^{01,r_1}_{p'q} \tilde{\rho}^{01,r_2}_{q'p}\\
  &\qquad\qquad\qquad\:
    + \left(\tilde{\kappa}^{10,r_1}_{pq}\right)^* \tilde{\kappa}^{01,r_2}_{p'q'}
  \Bigr),
  \label{eq:robust_2tdm}
\end{align}
where we have used $\lambda_{r_1 r_2} = \lambda_{r_2 r_1}$.
Computing the full \mbox{$2$-RTM} using Eq.~\eqref{eq:robust_2tdm} scales as $\mathcal{O}(M^5)$.
In practice, however, we never construct
the full $2$-RTM. To compute the matrix element of a two-body operator,
we can directly contract its parameters (\textit{e.g.,} two-electron integrals)
with the factorized form of the $2$-RTM in Eq.~\eqref{eq:robust_2tdm}.
For example, the tensor contraction between the two-electron integrals and the $2$-RTM
should scale as $\mathcal{O}(M^3)$ to $\mathcal{O}(M^4)$ after exploiting locality
using standard techniques.
\cite{Haser1989DSCF, Strout1995Scaling, Challacombe1997LinearScaling}

\section{Low-Scaling Version}
\label{sec:ls_robust_wick}

We may further reduce the scaling by
enforcing the cancellation only of the smallest elements in $s$.
We partition $\Gamma$
into singular ($S$) and regular ($R$) parts according to
\begin{equation}
  \label{eq:gamma_partitioned}
  \Gamma
  = \Gamma^S + \Gamma^R
  = \sum_{r \in S} s_r^{-1} \tilde{\Gamma}^r + \Gamma^R,
\end{equation}
where $S = \left\{s_r < \epsilon \mid 1 \leq r \leq M, \, \epsilon > 0\right\}$.
The value of $\epsilon$ should be so chosen that $\Gamma^R$
is well-conditioned while the size of $S$ is small.
Similarly, we define
$\rho^{01,R}$, $\kappa^{01,R}$, and $\kappa^{10,R}$
as the regular parts of $\rho^{01}$, $\kappa^{01}$, and $\kappa^{10}$,
respectively.
By the properties of the hyper-Pfaffian,
we can now establish a low-scaling robust Wick's theorem,
which states that
\begin{widetext}
\begin{align}
  \nonumber
  \braket{\Phi_0|\gamma^1 \gamma^2 \cdots \gamma^{2d}|\Phi_1}
  &= \zeta \lambda^R
  \Biggl(
    \lambda^S \pf(\Gamma^R)
    + (2d-1)!!\, {d \choose 1} \sum_{r_1 \in S} \lambda^S_{r_1} \,
    \hpf\left(\tilde{\Gamma}^{r_1} \otimes \left(\Gamma^R\right)^{\otimes (d-1)}\right)\\
    \nonumber
    &\qquad\qquad + (2d-1)!!\, {d \choose 2} \sum_{r_1, r_2 \in S} \lambda^S_{r_1 r_2} \,
    \hpf\left(\tilde{\Gamma}^{r_1} \otimes \tilde{\Gamma}^{r_2} \otimes \left(\Gamma^R\right)^{\otimes (d-2)}\right)\\
    &\qquad\qquad
    + \cdots
    + (2d-1)!! \sum_{r_1, r_2, \cdots, r_d \in S} \lambda^S_{r_1 r_2 \cdots r_d} \,
    \hpf(\tilde{\Gamma}^{r_1} \otimes \tilde{\Gamma}^{r_2} \otimes \cdots \otimes \tilde{\Gamma}^{r_d})
  \Biggr),
  \label{eq:ls_robust_wick}
\end{align}
\end{widetext}
where
\begin{subequations}
\begin{equation}
  \lambda^S = \prod_{r \in S} s_r,
  \quad
  \lambda^R = \prod_{r \notin S} s_r,
\end{equation}
\begin{equation}
  \lambda^S_{r_1 \cdots r_d} =
  \begin{cases}
    \displaystyle{\prod_{S \ni r \neq r_1, \cdots, r_d} s_r} &
    (r_1, \cdots, r_d \;\text{are distinct})\\
    \quad\quad\; 0 & (\text{otherwise})
  \end{cases}
\end{equation}
\end{subequations}
and we have used
\begin{equation}
  \hpf(\underbrace{\Gamma^R \otimes \cdots \otimes \Gamma^R \vphantom{\frac{a}{b}}}_d)
  = \hpf\left(\left(\Gamma^R\right)^{\otimes d}\right)
  = \frac{\pf(\Gamma^R)}{(2d-1)!!}
\end{equation}
If $S$ is empty, Eq.~\eqref{eq:ls_robust_wick}
reduces to the nonorthogonal Wick's theorem
in Eq.~\eqref{eq:nonorthogonal_wick}.

Applying the low-scaling robust Wick's theorem to \mbox{$2$-RTM}, we have
\begin{align}
  \nonumber
  &\braket{\Phi_0|
    \contraction[2ex]{}{c}{{}^\dag_p c^\dag_q c_{q'}}{c}
    \contraction[1.2ex]{c^\dag_p}{c}{{}^\dag_q}{c}
    c^\dag_p c^\dag_q c_{q'} c_{p'}
  |\Phi_1}
  =\\
  \nonumber
  &\quad \zeta \lambda^R \Biggl(
    \lambda^S \rho^{01,R}_{pq} \rho^{01,R}_{p'q'}
    + \sum_{r_1 \in S} \lambda^S_{r_1} \tilde{\rho}^{01,r_1}_{pq} \rho^{01,R}_{p'q'}\\
    &\qquad\quad
    + \sum_{r_1 \in S} \lambda^S_{r_1} \rho^{01,R}_{pq} \tilde{\rho}^{01,r_1}_{p'q'}
    + \sum_{r_1, r_2 \in S} \lambda^S_{r_1 r_2} \tilde{\rho}^{01,r_1}_{pq} \tilde{\rho}^{01,r_2}_{p'q'}
  \Biggr)
\end{align}
and similarly for 
$\braket{\Phi_0|
  \contraction[2ex]{}{c}{{}^\dag_p c^\dag_q}{c}
  \contraction[1.2ex]{c^\dag_p}{c}{{}^\dag_q c_{q'}}{c}
  c^\dag_p c^\dag_q c_{q'} c_{p'}
|\Phi_1}$
and
$\braket{\Phi_0|
  \contraction[1.2ex]{}{c}{{}^\dag_p}{c}
  \contraction[1.2ex]{c^\dag_p c^\dag_q}{c}{{}_{q'}}{c}
  c^\dag_p c^\dag_q c_{q'} c_{p'}
|\Phi_1}$.
Each term in the above expression stays factorized, facilitating
low-scaling tensor contractions with the electron integrals.
Besides, the indices $r_1$, $r_2$ now only run over
a small subset of $s$ elements. As a result, the computational cost
of the matrix element of a two-body operator
scales the same as that of HF, \textit{i.e.}, $\mathcal{O}(M^2)$ to $\mathcal{O}(M^3)$,
albeit with a larger prefactor.
For $d$-RTM with $2 < d \ll M$, the reduction of computational scaling
from Eq.~\eqref{eq:robust_wick} to Eq.~\eqref{eq:ls_robust_wick}
is even more significant.

We can now generalize the Slater--Condon rules to the HFB case
using the low-scaling robust Wick's theorem.
Let $\mu$ be the number of $s$ elements that are strictly zero.
Loosely speaking, $\mu$ equals the smallest number of levels
we need to block in both $\ket{\Phi_0}$ and $\ket{\Phi_1}$
in order to make them nonorthogonal.
The blocking should be done in a biorthogonal basis
that simultaneously brings both $\ket{\Phi_0}$ and $\ket{\Phi_1}$
to the canonical form of the Bardeen--Cooper--Schrieffer (BCS) ansatz.
\cite{Burzynski1995Quadruple, Doenau1998Canonical, Dobaczewski2000GBMZ}
This biorthogonal basis is generally nonunitary and does not always exist as noted in
Ref.~\onlinecite{Dobaczewski2000GBMZ}.
It is readily seen from Eq.~\eqref{eq:ls_robust_wick} that
the $\lambda^S$ term vanishes for $\mu = 1$,
the $\lambda^S$ and $\lambda^S_{r_1}$ terms vanish for $\mu = 2$,
and so forth.
These results resemble the generalized Slater--Condon rules for HF matrix elements,
\cite{Verbeek1991GSC, Thom2009NOCI, Burton2021Wick}
in which case $\mu$ corresponds to the smallest number of biorthogonal occupied orbital
pairs that need to be emptied in order to make the HF wavefunctions nonorthogonal.
Here the biorthogonalization is guaranteed by L{\"o}wdin pairing and realized by SVD.
\cite{Amos1961SingleDeterminant, Karadakov1985Pairing, Mayer2010Pairing}
We should point out that, as opposed to the generalized Slater--Condon rules,
our Eq.~\eqref{eq:ls_robust_wick}
also ensures numerical stability when some elements of $s$
are small but not strictly zero.

\section{Robust Normalization}
\label{sec:normalization}

So far we have considered unnormalized HFB wavefunctions $\ket{\Phi_j}$ for $j = 0, 1$ 
as defined in Eq.~\eqref{eq:hfb}.
The corresponding normalized HFB wavefunction is\cite{Bertsch2012Pfaffian}
\begin{equation}
  \label{eq:normalized_hfb}
  \ket{\Phi'_j}
  = \frac{\det(C^j)}{\prod_{p=1}^{m^j_P} v^j_p} \ket{\Phi_j},
\end{equation}
where $C^j$ and $v^j$ come from the Bloch--Messiah decomposition,\cite{Bloch1962Canonical}
\begin{align}
  U^j = D^j \bar{U}^j C^j, \quad
  V^j = \left(D^j\right)^* \bar{V}^j C^j
\end{align}
with $D^j$ and $C^j$ being unitary and
\begin{subequations}
\label{eq:ubarvbar}
\begin{align}
  \bar{U}^j
  & = \begin{pmatrix}
    0_{M^j_C \times M^j_C} & & \\
    & -\displaystyle{\bigoplus_{p=1}^{m^j_P}} i \sigma_y u^j_p & \\
    & & I_{M^j_V}
  \end{pmatrix},\\
  \bar{V}^j
  &= \begin{pmatrix}
    I_{M^j_C} & & \\
    & \displaystyle{\bigoplus_{p=1}^{m^j_P}} I_2 v^j_p & \\
    & & 0_{M^j_V \times M^j_V}
  \end{pmatrix}.
\end{align}
\end{subequations}
Here $I_k$ denotes the $k \times k$ identity matrix, and
$\sigma_y$ denotes the $y$-component of the $2 \times 2$ Pauli matrices;
$u^j_p, v^j_p$ are positive real numbers and satisfy
$\left(u^j_p\right)^2 + \left(v^j_p\right)^2 = 1$.
Note that we use an unconventional definition of $(\bar{U}^j, \bar{V}^j)$ in
Eq.~\eqref{eq:ubarvbar}, which is physically inconsequential
but will
ease the notation in what follows. The Bloch--Messiah decomposition
implies that, in the canonical basis of $\ket{\Phi_j}$ with orbital coefficients $D^j$,
we have $M^j_C$ fully occupied core orbitals, $2 m^j_P$ paired fractionally occupied
orbitals, and $M^j_V$ unoccupied or virtual orbitals, with
$M^j_C + 2 m^j_P + M^j_V = M$.

The normalization factor in Eq.~\eqref{eq:normalized_hfb} is generally unbounded,
which may introduce additional numerical challenges.
To overcome this problem, we define a set of unnormalized quasiparticles by an unnormalized
Bogoliubov transformation:
\begin{equation}
  \beta^{j\prime}_p
  = \sum_{q=1}^M \left(\left(U^{j\prime}_{qp}\right)^* c_q + \left(V^{j\prime}_{qp}\right)^* c^\dag_q\right)
\end{equation}
for $1 \leq p \leq M^j_C + 2 m^j_P$, with
\begin{align}
\label{eq:defuvprime}
  U^{j\prime}
  = D^{j\prime} \bar{U}^{j\prime} \left(\bar{V}^{j\prime}\right)^{-\frac{1}{2}}, \quad
  V^{j\prime}
  = \left(D^{j\prime}\right)^* \left(\bar{V}^{j\prime}\right)^{\frac{1}{2}},
\end{align}
where $D^{j\prime}$ is built from the first $M^j_C + 2 m^j_P$ columns of $D^j$ 
as suggested in Ref.~\onlinecite{Robledo2011Technical}, and
\begin{subequations}
\label{eq:defuvbarprime}
\begin{align}
  \bar{U}^{j\prime}
  &= \begin{pmatrix}
    0_{M^j_C \times M^j_C} & \\
    & -\displaystyle{\bigoplus_{p=1}^{m^j_P}} i \sigma_y u^j_p
  \end{pmatrix},\\
  \bar{V}^{j\prime}
  &= \begin{pmatrix}
    I_{M^j_C} & \\
    & \displaystyle{\bigoplus_{p=1}^{m^j_P}} I_2 v^j_p
  \end{pmatrix}.
\end{align}
\end{subequations}
It is readily shown that the normalized HFB wavefunction $\ket{\Phi'_j}$ is
the vacuum state of these unnormalized quasiparticles, \textit{i.e.},
\begin{equation}
  \label{eq:normalized_hfb2}
  \ket{\Phi'_j}
  = \prod_{p=1}^{M^j_C + 2 m^j_P} \beta^{j\prime}_p \ket{-}.
\end{equation}
Furthermore,
\begin{subequations}
\label{eq:defsprime}
\begin{align}
  \mathcal{S}'
  &= \begin{pmatrix}
    \left(V^{0\prime}\right)^\T U^{0\prime} &
    \left(V^{0\prime}\right)^\T \left(V^{1\prime}\right)^*\\
    -\left(V^{1\prime}\right)^\dag V^{0\prime} &
    \left(U^{1\prime}\right)^\dag \left(V^{1\prime}\right)^*
  \end{pmatrix}\\
  &= \begin{pmatrix}
    \bar{U}^{0\prime} &
    \left(V^{0\prime}\right)^\T \left(V^{1\prime}\right)^*\\
    -\left(V^{1\prime}\right)^\dag V^{0\prime} &
    -\bar{U}^{1\prime}
  \end{pmatrix},
\end{align}
\end{subequations}
whose elements are bounded by $1$ in absolute value
even though $U^{j\prime}$ is generally unbounded.
Therefore, we no longer have to deal with the problematic normalization factor
in Eq.~\eqref{eq:normalized_hfb} as long as we compute the overlap
between normalized HFB wavefunctions using
\begin{equation}
  \braket{\Phi'_0|\Phi'_1}
  = (-1)^{(M^0_C + 2 m^0_P)(M^0_C + 2 m^0_P - 1)/2} \pf(\mathcal{S}').
\end{equation}

This procedure for computing the overlap is similar to the one proposed in
Ref.~\onlinecite{Carlsson2021Overlap};
however, we emphasize that our formalism,
as opposed to those in Refs.~\onlinecite{Bertsch2012Pfaffian} and \onlinecite{Carlsson2021Overlap},
treats even- and odd-particle systems
on the same footing, since the number of core orbitals $M^j_C$ from the Bloch--Messiah decomposition
is not restricted to be even.
Other formalisms for equally accounting for even and odd number parities have been suggested,
without addressing the numerical pitfalls of the unbounded normalization factors.
\cite{Avez2012Overlap, Jin2022MPS, Yang2023Projected}

More importantly, our robust normalization procedure for computing the overlap
is compatible with the robust Wick's theorem and its low-scaling version.
Namely, Eqs.~\eqref{eq:robust_wick} and \eqref{eq:ls_robust_wick} remain valid for
the $d$-RTM between the normalized $\ket{\Phi'_0}$ and $\ket{\Phi'_1}$ 
after formally replacing $(U^j, V^j)$ with $(U^{j\prime}, V^{j\prime})$,
or equivalently replacing $\mathcal{S}$ with $\mathcal{S}'$ and
$\mathcal{V}$ with
\begin{equation}
  \mathcal{V}'
  = \begin{pmatrix}
    0   & -\left(V^{1\prime}\right)^*\\
    V^{0\prime} & 0
  \end{pmatrix}.
\end{equation}
This correspondence is apparent from Eq.~\eqref{eq:normal_ordered},
whose derivation does not rely on the normalization of the Bogoliubov transformations.

Finally, we note in passing that, when $m^j_P = 0$, $\ket{\Phi'_j}$ reduces to a HF wavefunction,
so the robust Wick's theorem and its low-scaling variant can also be used to compute
HF reduced transition matrices (see the supplementary material).
All pairing-type contractions are removed in this case due to vanishing $\kappa^{01}$
and $\kappa^{10}$.

\begin{figure*}[t]
\includegraphics[width=\columnwidth]{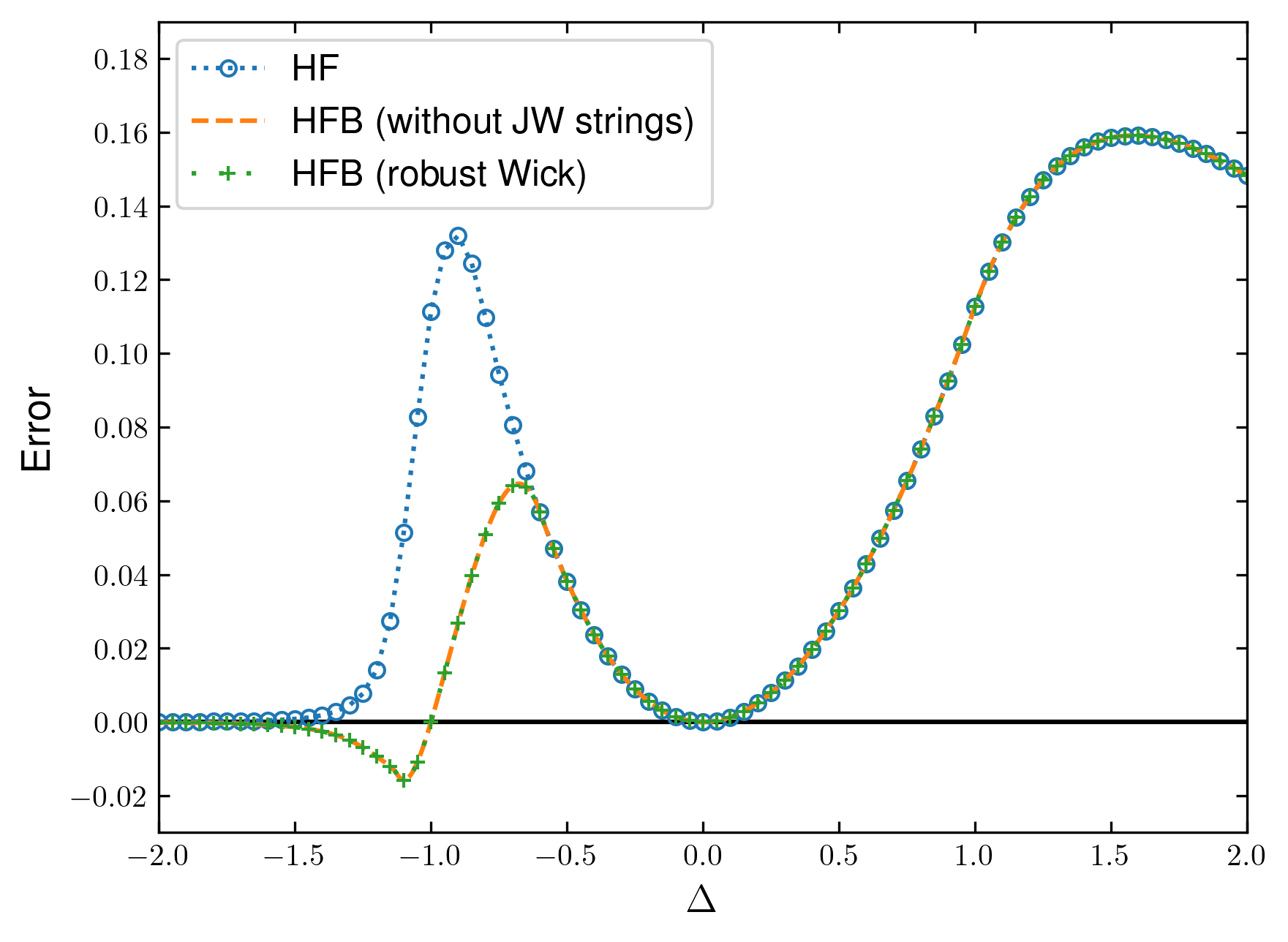}
\hfill
\includegraphics[width=\columnwidth]{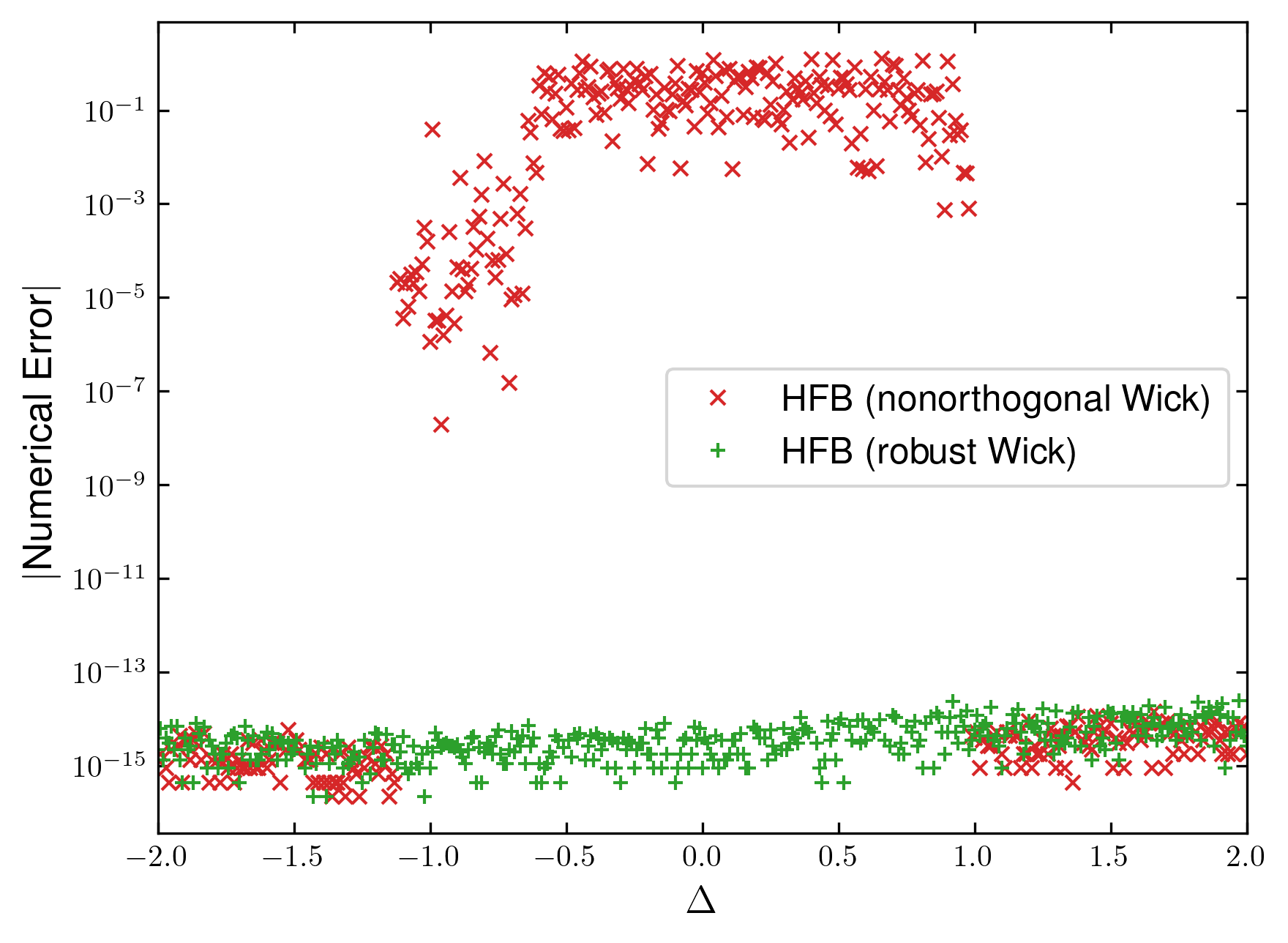}
\caption{Energy errors for the 8-site XXZ chain with $\braket{S^z} = 0$.
Left panel: Errors with respect to the full configuration interaction (FCI) energy. Hartree--Fock (HF) results are plotted to show symmetry breaking. Right panel: Numerical errors in absolute value with respect to the Hartree--Fock--Bogoliubov (HFB) energy computed without Jordan--Wigner (JW) strings according to Eq.~\eqref{eq:1djw_matrix_element}.}
\label{fig:robust_vs_naive}
\end{figure*}

\section{Numerical Example}

We here present a prototypical example to demonstrate the power of our proposed robust Wick's theorem. The Heisenberg XXZ Hamiltonian can be treated as a fermionic system after Jordan--Wigner (JW) transformation. The claim to fame of this mapping is that in 1D at $\Delta=0$, a maximally degenerate point, the transformed Hamiltonian becomes a free fermion system readily solvable by mean-field HF theory. Here, $\Delta$ is the anisotropy parameter. We refer the reader to Ref.~\onlinecite{Henderson2022Duality} for details of the transformed Hamiltonian. For $\Delta<0$, number symmetry breaks spontaneously, necessitating HFB solutions. Matrix elements of the form 
\begin{equation}
  \label{eq:jw_matrix_element}
  \braket{\Phi|c^\dag_p c_q|\Phi^{pq}}
  = \braket{\Phi|c^\dag_p c_q \phi_p \phi_q|\Phi}
\end{equation}
need to be computed for $1 \leq p < q \leq M$, where the JW string
\begin{equation}
  \phi_p 
  = \prod_{q < p} (1 - 2 n_q)
  = \exp(i \pi \sum_{q < p} n_q)
\end{equation}
is a Thouless rotation acting on $\ket{\Phi}$ with $n_q = c^\dag_q c_q$. 
The overlap
\begin{equation}
  \braket{\Phi|\Phi^{pq}}
  = \braket{\Phi|\phi_p \phi_q|\Phi}
  = \braket{\Phi|\prod_{r=p}^{q-1} (1 - 2 n_r)|\Phi}
\end{equation}
vanishes when $\braket{\Phi|n_r|\Phi} = 1/2$ for $p \leq r < q$, a situation we observe across the critical
$-1 \lesssim \Delta \lesssim 1$ region of the 1D phase diagram. On the other hand, because XXZ possess only
nearest neighbor interactions, one can avoid dealing with JW strings all together and analytically show that
\begin{equation}
  \label{eq:1djw_matrix_element}
  \braket{\Phi|c^\dag_p c_{p+1}|\Phi^{p,p+1}}
  = -\braket{\Phi|c^\dag_p c_{p+1}|\Phi}.
\end{equation}
The right hand side of Eq.~\eqref{eq:1djw_matrix_element} is numerically well-posed because $\braket{\Phi|\Phi} = 1$,
after applying the normalization procedure in Sec.~\ref{sec:normalization}. Therefore, we have two analytically equivalent
forms of evaluating the same matrix elements. Comparing numerical results from Eq.~\eqref{eq:jw_matrix_element},
which generates orthogonal HFB states, and Equation\eqref{eq:1djw_matrix_element}, we can measure the round-off errors
due to vanishing $\braket{\Phi|\Phi^{pq}}$ and verify that our robust Wick's theorem eliminates them.

Consider a JW-transformed 8-site XXZ chain with open boundary conditions. We constrain $\braket{\Phi|\sum_p n_p|\Phi} = M/2$,
which corresponds to $\braket{S^z} = 0$. We find that $\braket{\Phi|\Phi^{p,p+1}} = 0$ for all $1 \leq p < M$ with $-1.12 \leq \Delta \leq 0.98$.
We observe equal site occupations ($\braket{S^z_p} = 0$ in the original spin representation) within this range of $\Delta$ values,
consistent with our previous calculations using a spin antisymmetrized geminal power ansatz.\cite{Liu2023SAGP}
As shown in Fig.~\ref{fig:robust_vs_naive}, the HFB energies obtained from Eq.~\eqref{eq:jw_matrix_element} using
the nonorthogonal Wick's theorem indeed suffer from severe numerical error in this regime. In contrast,
the robust Wick's theorem, yields energies on top of the exact results computed without JW strings using Eq.~\eqref{eq:1djw_matrix_element}.

\section{Concluding Remarks}

In quantum mechanical calculations with nonorthogonal basis,
negligible overlaps do not necessarily imply negligible interactions,
as one might erroneously infer from the nonorthogonal Wick's
theorem in Eq.~\eqref{eq:nonorthogonal_wick}, where the matrix element
is proportional to the overlap $\braket{\Phi_0|\Phi_1}$. The fallacy originates
from the fact that
Eq.~\eqref{eq:nonorthogonal_wick} is ill-defined when $\braket{\Phi_0|\Phi_1} = 0$,
as shown by our formal and numerical analyses.
Accurate computation of the matrix element requires evaluating the limit of
Eq.~\eqref{eq:nonorthogonal_wick} as $\braket{\Phi_0|\Phi_1} \to 0$,
which is always well-defined and can be nonzero in general.

The idea of taking the limit before numerical evaluation is 
a recurring theme of this paper. It underlies how we resolve the two
major limitations of the standard formulation of the nonorthogonal Wick's theorem:
\begin{itemize}
  \item First, when $\braket{\Phi_0|\Phi_1}$ vanishes, the Pfaffian in
    Eq.~\eqref{eq:nonorthogonal_wick} develop poles even after
    self-interaction is eliminated. We have addressed this issue in
    Sec.~\ref{sec:robust_wick} by cancelling these poles with the zeros
    of the overlap. The resulting robust formulation of Wick's theorem
    is universally applicable and numerically stable for computing matrix elements between HFB
    wavefunctions. And it remains amenable
    to efficient implementation as shown in Sec.~\ref{sec:ls_robust_wick}.

  \item Second, the normalization factors for $\ket{\Phi_0}$ and $\ket{\Phi_1}$
    may become divergingly large. We have addressed this issue in
    Sec.~\ref{sec:normalization} by introducing a robust normalization procedure.
    This is
    achieved by formally viewing normalized HFB wavefunctions as
    vacuum states of unnormalized quasiparticles.
    Although the coefficients
    defining these unnormalized quasiparticles are unbounded in general, the
    $\mathcal{S}'$ and $\mathcal{V}'$ matrices entering the final expressions
    for the matrix elements remain bounded.
\end{itemize}

We expect the present work to have a broad impact in quantum many-body theories,
facilitating future developments of methods based on quasiparticle vacuum states.
One promising area of application is HFB-based NOCI or generator coordinate method (GCM) in general.
\cite{HillWheeler1953, Peierls1957Collective}
Not only do these methods recover important ground-state correlations,
\cite{Scuseria2011PQT, Uemura2015AGPCI}
but they
are also expected to be suitable for describing excited states, particularly those of charge transfer or 
double excitation character.
\cite{Sundstrom2014NOCI, Nite2019saResHF, Mahler2021ooNOCI}
Moreover, PHFB is closely related to geminal theories.
\cite{Scuseria2011PQT, Johnson2013101SizeConsistent, Johnson2017OpenShell, Cassam2023Geminal}
Our formalism can be applied to compute matrix elements between number-projected HFB (NHFB) or
antisymmetrized geminal power (AGP) states, especially for ones that do not share the same set of
canonical orbitals.
Another potential application is correlated methods using HFB or PHFB as the reference state.
\cite{Duguet2017PCC, Qiu2019PBCSCC, Baran2021Variational, Henderson2020CorrelatingAGP, 
Khamoshi2021JCC}
Reduced transition matrices encountered in these methods are typically of high order, for which the
low-scaling version of our robust Wick's theorem becomes useful.


%
%

%

\section*{Supplementary Material}

The supplementary material contains definitions of the Pfaffian and hyper-Pfaffian
along with detailed proofs of equations and theorems in the main text. 
We also show therein how our robust formulation of Wick's theorem reduces to
the RTM expressions in Ref.~\onlinecite{Koch1993LinearSuperposition}
in the Hartree--Fock limit.

\begin{acknowledgments}

This work was supported by the U.S. National Science Foundation under Grant No.~CHE-2153820.
G.E.S.~is a Welch Foundation Chair (Grant No.~C-0036).
G.P.C.~thanks Thomas M.~Henderson for helpful discussions and useful comments.

\end{acknowledgments}

\bibliography{ref}

\end{document}


\renewcommand{\thesubsection}{\thesection.\arabic{subsection}}

\maketitle

\section{Definitions}%

We follow our notation and definitions in the main text.
Additional definitions are introduced below to facilitate
proofs and discussions in subsequent sections.

\subsection{Pfaffian and Hyper-Pfaffian}

\begin{definition}[Pfaffian]
  The Pfaffian of a $2d \times 2d$ antisymmetric matrix $Z$ is defined as
  \begin{equation}
    \label{eq:pfaffian}
    \pf(Z)
    = \frac{1}{2^d d!} \sum_{\sigma \in S_{2d}} \text{sgn}(\sigma)
    \prod_{p=1}^d Z_{\sigma(2p-1) \sigma(2p)}.
  \end{equation}
  For a general antisymmetric matrix of dimension $n \times n$, we set its Pfaffian to $0$
  if $n$ is odd.
\end{definition}

\begin{definition}[Hyper-Pfaffian]
  Let $\mathcal{Z}^{[k]}$ be a $k$-index tensor of dimension $n^k$ that satisfies
  $n = k l$ for some positive integer $l$ and
  \begin{equation}
    \mathcal{Z}^{[k]}_{p_1, p_2, \cdots, p_k}
    = \text{sgn}(\sigma)
    \mathcal{Z}^{[k]}_{\sigma(p_1), \sigma(p_2), \cdots, \sigma(p_k)}
  \end{equation}
  for $\forall \sigma \in S_k$.
  The hyper-Pfaffian of $\mathcal{Z}^{[k]}$ is defined as
  \begin{equation}
    \label{eq:hyperpfaffian}
    \hpf(\mathcal{Z}^{[k]})
    = \frac{1}{(k!)^l l!} \sum_{\sigma \in S_n} \text{sgn}(\sigma)
    \prod_{p=1}^l \mathcal{Z}^{[k]}_{\sigma(pk-k+1),\sigma((pk-k+2),\cdots,\sigma(pk)}.
  \end{equation}
\end{definition}

Note that this definition is consistent with the one by Luque and Thibon \cite{Luque2002Pfaffian}
and differ from Barvinok's definition \cite{Barvinok1995Hyperpfaffian} by a factor of $(k!)^{-l}$.
It is readily seen that our definition of the hyper-Pfaffian reduces to that of the Pfaffian when
$k = 2$.

These definitions of the Pfaffian and the hyper-Pfaffian can be straightforwardly extended to
any unstructured square matrix $Z$ and $n \times n \times \cdots \times n$ tensor $\mathcal{Z}^{[k]}$,
respectively. In particular, for the $2d$-index tensor 
$\tilde{\Gamma}^{r_1} \otimes \tilde{\Gamma}^{r_2} \otimes \cdots \otimes \tilde{\Gamma}^{r_d}$
in Equation~(28), we have $k = n = 2d$, $l = 1$, and therefore
\begin{subequations}
\label{eq:hyperpfaffian_special}
\begin{align}
  \hpf(\tilde{\Gamma}^{r_1} \otimes \tilde{\Gamma}^{r_2} \otimes \cdots \otimes \tilde{\Gamma}^{r_d})
  &= \frac{1}{(2d)!} \sum_{\sigma \in S_{2d}} \text{sgn}(\sigma)
  \left[
    \tilde{\Gamma}^{r_1} \otimes \tilde{\Gamma}^{r_2} \otimes \cdots \otimes \tilde{\Gamma}^{r_d}
  \right]_{\sigma(1), \sigma(2), \cdots, \sigma(2d)}\\
  &= \frac{1}{(2d)!} \sum_{\sigma \in S_{2d}} \text{sgn}(\sigma)
  \tilde{\Gamma}^{r_1}_{\sigma(1) \sigma(2)}
  \tilde{\Gamma}^{r_2}_{\sigma(3) \sigma(4)}
  \cdots
  \tilde{\Gamma}^{r_d}_{\sigma(2d-1) \sigma(2d)},
\end{align}
\end{subequations}
assuming the following tensor product convention:
\begin{equation}
  \left[
    \mathcal{Z}^{[k]} \otimes \tilde{\mathcal{Z}}^{[\tilde{k}]}
  \right]_{p_1, p_2, \cdots, p_k, q_1, q_2, \cdots, q_{\tilde{k}}}
  = \mathcal{Z}^{[k]}_{p_1, p_2, \cdots, p_k}
  \tilde{\mathcal{Z}}^{[\tilde{k}]}_{q_1, q_2, \cdots, q_{\tilde{k}}}.
\end{equation}

\subsection{Miscellaneous}
\label{sub:misc}

The unitary matrix
\begin{equation}
  \mathcal{W}^j
  = \begin{pmatrix}
    U^j & V^{j*}\\
    V^j & U^{j*}
  \end{pmatrix}
\end{equation}
represents the Bogoliubov transformation from $(c^\dag, c)$ to $(\beta^{j\dag}, \beta^j)$,
for $j = 0,1$. It is readily seen that the Bogoliubov transformation from $(\beta^{0\dag}, \beta^0)$ to
$(\beta^{1\dag}, \beta^1)$ is represented by
\begin{equation}
  \label{eq:trans0to1}
  \mathcal{T}
  = \mathcal{W}^{0\dag} \mathcal{W}^1
  = \begin{pmatrix}
    \Xi & \Lambda^*\\
    \Lambda & \Xi^*
  \end{pmatrix},
\end{equation}
where
\begin{subequations}
\begin{align}
  \label{eq:def_xi}
  \Xi
  &= U^{0\dag} U^1 + V^{0\dag} V^1,\\
  \label{eq:def_lambda}
  \Lambda
  &= U^{0\,\T} V^1 + V^{0\,\T} U^1.
\end{align}
\end{subequations}
Namely,
\begin{align}
  \label{eq:beta0tobeta1}
  \beta^1_p
  = \sum_{q=1}^M \left(\Xi^*_{qp} \beta^0_q + \Lambda^*_{qp} \beta^{0\dag}_q\right).
\end{align} 

We also define
\begin{equation}
  \tilde{\mathcal{S}}
  = \mathcal{J} \mathcal{S}
  = \begin{pmatrix}
    -V^{1\dag} V^0 & U^{1\dag} V^{1*}\\
    V^{0\,\T} U^0  & V^{0\,\T} V^{1*}
  \end{pmatrix},
\end{equation}
where
\begin{equation}
  \mathcal{J} =
  \begin{pmatrix}
    0   & I_M\\
    I_M & 0
  \end{pmatrix}.
\end{equation}

\section{Proof of Equation~(17)}%

Under the assumption that $V^j$ is non-singular for $j = 0, 1$ and $\braket{\Phi_0|\Phi_1} \neq 0$,
we first establish the following lemmas.

\begin{lemma}
  $\Xi$, defined in Equation~\eqref{eq:def_xi}, is non-singular.
\end{lemma}
\begin{proof}
  According to Ref.~\cite{Bertsch2012Pfaffian}, and using the non-singularity of $V^{1\dag} V^0$,
  we have
  \begin{subequations}
    \begin{align}
      \left(\braket{\Phi_0|\Phi_1}\right)^2
      &= \pf(\mathcal{S})^2 = \det(\mathcal{S}) = \det(\mathcal{J}) \det(\tilde{\mathcal{S}})\\
      &= (-1)^M \det
      \begin{pmatrix}
        -V^{1\dag} V^0 & U^{1\dag} V^{1*}\\
        V^{0\,\T} U^0  & V^{0\,\T} V^{1*}
      \end{pmatrix}\\
      &= (-1)^M \det(-V^{1\dag} V^0) \det(\tilde{\mathcal{S}}/(-V^{1\dag} V^0))\\
      &= \det(V^{1\dag} V^0) \det(\tilde{\mathcal{S}}/(-V^{1\dag} V^0)),
    \end{align}
  \end{subequations}
  where the Schur complement of $-V^{1\dag} V^0$ of the supermatrix $\tilde{\mathcal{S}}$
  is written as
  \begin{subequations}
  \label{eq:s_schur1}
  \begin{align}
    \tilde{\mathcal{S}}/(-V^{1\dag} V^0)
    &= V^{0\,\T} V^{1*} - V^{0\,\T} U^0 (-V^{1\dag} V^0)^{-1} U^{1\dag} V^{1*}\\ 
    &= V^{0\,\T} V^{1*} + U^{0\,\T} V^0 (V^{1\dag} V^0)^{-1} V^{1\dag} U^{1*}\\ 
    &= \Xi^*.
  \end{align}
  \end{subequations}
  In the last equality, we have used $U^{j\,\T} V^j + V^{j\,\T} U^j = 0$, which is a consequence
  of the unitarity of $\mathcal{W}^j$.
  We therefore obtain
  \begin{equation}
    \label{eq:onishi}
    \left(\braket{\Phi_0|\Phi_1}\right)^2
    = \det(V^{1\dag} V^0) \det(\Xi^*).
  \end{equation}
  It follows from $\braket{\Phi_0|\Phi_1} \neq 0$ and the non-singularity of $V^0$ and $V^1$
  that $\Xi$ is also non-singular.
\end{proof}

Equation~\eqref{eq:onishi} is consistent with the Onishi formula \cite{Onishi1966Generator}.
Note that the sign ambiguity \cite{Robledo2009Sign, Mizusaki2018Sign} 
is irrelevant here since we have squared the overlap.
The additional $\det(V^{1\dag} V^0)$ factor originates from the fact that $\ket{\Phi_0}$
and $\ket{\Phi_1}$ are unnormalized.

\begin{lemma}
  \begin{subequations}
  \begin{align}
    \rho^{01}_{pq}
    &= \contraction[1.2ex]{}{c}{{}^\dag_q}{c}
    c^\dag_q c_p
    = \left[V^{1*} \left(\Xi^{-1}\right)^* V^{0\,\T}\right]_{pq},\\
    \label{eq:kappa01def}
    \kappa^{01}_{pq}
    &= \contraction[1.2ex]{}{c}{{}_q}{c}
    c_q c_p
    = \left[V^{1*} \left(\Xi^{-1}\right)^* U^{0\,\T}\right]_{pq},\\
    \label{eq:kappa10def}
    \kappa^{10*}_{pq}
    &= \contraction[1.2ex]{}{c}{{}^\dag_p}{c}
    c^\dag_p c^\dag_q
    = -\left[U^{1*} \left(\Xi^{-1}\right)^* V^{0\,\T}\right]_{pq}.
  \end{align}
  \end{subequations}
\end{lemma}

\begin{proof}
  These are well-established results\cite{Ring1980Nuclear}.
  Here we provide a proof just for completeness.

  Equation~\eqref{eq:beta0tobeta1} and the non-singularity of $\Xi$ imply
  \begin{equation}
    \beta^0_p = \sum_q
    \left(
      \left[\Xi^{-1}\right]^*_{qp} \beta^1_q
      - \left[\Lambda \Xi^{-1}\right]^*_{qp} \beta^{0\dag}_q
    \right).
  \end{equation}
  It follows that
  \begin{subequations}
  \begin{align}
    \braket{\Phi_0|c^\dag_q c_p|\Phi_1}
    &= \sum_{rs}
    \braket{\Phi_0|
      \left(U^{0*}_{qr} \beta^{0\dag}_r + V^0_{qr} \beta^0_r\right)
      \left(V^{1*}_{ps} \beta^{1\dag}_s + U^1_{ps} \beta^1_s\right)
    |\Phi_1}\\
    &= \sum_{rs} V^0_{qr} V^{1*}_{ps} \braket{\Phi_0|\beta^0_r \beta^{1\dag}_s|\Phi_1}\\
    &= \sum_{rst} V^0_{qr} V^{1*}_{ps} 
    \braket{\Phi_0|
      \left(
        \left[\Xi^{-1}\right]^*_{tr} \beta^1_t
        - \left[\Lambda \Xi^{-1}\right]^*_{tr} \beta^{0\dag}_t
      \right) \beta^{1\dag}_s
    |\Phi_1}\\
    &= \sum_{rst} V^{1*}_{ps} \left[\Xi^{-1}\right]^*_{tr} V^0_{qr}
    \braket{\Phi_0|\beta^1_t \beta^{1\dag}_s|\Phi_1}\\
    &= \sum_{rst} V^{1*}_{ps} \left[\Xi^{-1}\right]^*_{tr} V^0_{qr}\, \delta_{ts}
    \braket{\Phi_0|\Phi_1}\\
    &= \left[V^{1*} \left(\Xi^{-1}\right)^* V^{0\,\T}\right]_{pq}
    \braket{\Phi_0|\Phi_1},
  \end{align}
  \end{subequations}
  and therefore,
  \begin{equation}
    \rho^{01}_{pq}
    = \contraction[1.2ex]{}{c}{{}^\dag_q}{c}
    c^\dag_q c_p
    = \frac{\braket{\Phi_0|c^\dag_q c_p|\Phi_1}}{\braket{\Phi_0|\Phi_1}}
    = \left[V^{1*} \left(\Xi^{-1}\right)^* V^{0\,\T}\right]_{pq}.
  \end{equation}
  We can show Equations~\eqref{eq:kappa01def} and \eqref{eq:kappa10def} similarly.
  We note that $\kappa^{01}$ and $\kappa^{10}$ are antisymmetric by definition.
\end{proof}

We now proceed to show Equation~(17).

\begin{proof}
  Given that $V^0$ and $V^1$ are non-singular, one can verify that
  \begin{equation}
    \tilde{\mathcal{S}}/(V^{0\,\T} V^{1*})
    = - \Xi^\dag
  \end{equation}
  in analogous to Equation~\eqref{eq:s_schur1}.
  Moreover, by block LDU decomposition and the Sherman--Morrison--Woodbury identity, we have
  \begin{subequations}
  \begin{align}
    \tilde{\mathcal{S}}^{-1}
    &= \begin{pmatrix}
      -V^{1\dag} V^0 & U^{1\dag} V^{1*}\\
      V^{0\,\T} U^0  & V^{0\,\T} V^{1*}
    \end{pmatrix}^{-1}\\
    &= \begin{pmatrix}
      \left(\tilde{\mathcal{S}}/(V^{0\,\T} V^{1*})\right)^{-1} & 0\\
      0 & \left(\tilde{\mathcal{S}}/(-V^{1\dag} V^0)\right)^{-1}
    \end{pmatrix}
    \begin{pmatrix}
      I_M & -U^{1\dag} V^{1*} \left(V^{0\,\T} V^{1*}\right)^{-1}\\
      -V^{0\,\T} U^0 \left(-V^{1\dag} V^0\right)^{-1} & I_M
    \end{pmatrix}\\
    &= \begin{pmatrix}
      -\left(\Xi^{-1}\right)^\dag & 0\\
      0 & \left(\Xi^{-1}\right)^*
    \end{pmatrix}
    \begin{pmatrix}
      I_M & -U^{1\dag} \left(V^{0\,\T}\right)^{-1}\\
      -U^{0\,\T} \left(V^{1\dag}\right)^{-1} & I_M
    \end{pmatrix}\\
    &= \begin{pmatrix}
      -\left(\Xi^{-1}\right)^\dag &
      -\left(\Xi^{-1}\right)^\dag U^{1\dag} \left(V^{0\,\T}\right)^{-1}\\
      -\left(\Xi^{-1}\right)^* U^{0\,\T} \left(V^{1\dag}\right)^{-1} &
      \left(\Xi^{-1}\right)^*.
    \end{pmatrix}
  \end{align}
  \end{subequations}
  Therefore,
  \begin{align}
    \mathcal{S}^{-1}
    = \tilde{\mathcal{S}}^{-1} \mathcal{J}^{-1}
    = \begin{pmatrix}
      -\left(\Xi^{-1}\right)^\dag U^{1\dag} \left(V^{0\,\T}\right)^{-1} &
      -\left(\Xi^{-1}\right)^\dag\\
      \left(\Xi^{-1}\right)^* &
      -\left(\Xi^{-1}\right)^* U^{0\,\T} \left(V^{1\dag}\right)^{-1}
    \end{pmatrix},
  \end{align}
  and
  \begin{align}
    \mathcal{K}
    = \mathcal{V} \mathcal{S}^{-1} \mathcal{V}^\T
    = \begin{pmatrix}
      - V^{1*} \left(\Xi^{-1}\right)^* U^{0\,\T} &
      - V^{1*} \left(\Xi^{-1}\right)^* V^{0\,\T}\\
      V^0 \left(\Xi^{-1}\right)^\dag V^{1\dag} &
      V^0 \left(\Xi^{-1}\right)^\dag U^{1\dag}
    \end{pmatrix}
    = \begin{pmatrix}
      -\kappa^{01} & -\rho^{01}\\
      \rho^{01\,\T} & \kappa^{10*}
    \end{pmatrix}.
  \end{align}
\end{proof}

\section{Proof of Equation~(18)}%

\begin{proof}
  By Equations~(12), (16), and (17), we have
  \begin{equation}
    \mathcal{M}/\mathcal{S}
    = \Gamma^{(-)} + \mathcal{C}^\dag \mathcal{V} \mathcal{S}^{-1} \mathcal{V}^\T \mathcal{C}^*
    = \Gamma^{(-)} + \mathcal{C}^\dag \mathcal{K} \mathcal{C}^*.
  \end{equation}
  Therefore, for $1 \leq k < l \leq 2d$,
  \begin{equation}
    \Gamma^{(-)}_{kl}
    = \sum_{pq} A^*_{pk} B^*_{ql}\braket{-|c_p c^\dag_q|-}
    = \sum_p A^*_{pk} B^*_{pl}
  \end{equation}
  and
  \begin{subequations}
  \begin{align}
    \left[\mathcal{M}/\mathcal{S}\right]_{kl}
    &= \sum_p A^*_{pk} B^*_{pl}
    + \sum_{pq}
    \begin{pmatrix}
      A^*_{pk} & B^*_{pk}
    \end{pmatrix}
    \begin{pmatrix}
      -\kappa^{01}_{pq} & -\rho^{01}_{pq}\\
      \rho^{01}_{qp} & \kappa^{10*}_{pq}
    \end{pmatrix}
    \begin{pmatrix}
      A^*_{ql} \\ B^*_{ql}
    \end{pmatrix}\\
    &= \sum_{pq}
    \left(
      - A^*_{pk} A^*_{ql} \kappa^{01}_{pq}
      + A^*_{pk} B^*_{ql} (\delta_{pq} - \rho^{01}_{pq}) 
      + B^*_{pk} A^*_{ql} \rho^{01}_{qp}
      + B^*_{pk} B^*_{ql} \kappa^{10*}_{pq}
    \right)\\
    &= \sum_{pq}
    \frac{
      \braket{\Phi_0|
        \left(A^*_{pk} c_p + B^*_{pk} c^\dag_p\right)
        \left(A^*_{ql} c_q + B^*_{ql} c^\dag_q\right)
      |\Phi_1}
    }{\braket{\Phi_0|\Phi_1}}\\
    &= \frac{\braket{\Phi_0|\gamma^k \gamma^l|\Phi_1}}{\braket{\Phi_0|\Phi_1}}\\
    &= \Gamma_{kl}.
  \end{align}
  \end{subequations}
  By their antisymmetry, we have $\mathcal{M}/\mathcal{S} = \Gamma$.
\end{proof}

\section{Proof of the Robust Wick's Theorem}%

\begin{lemma}
  \label{lemma:cancellation}
  For the antisymmetric matrix $\tilde{\Gamma}^r$ defined in Equation~(27), we have
  \begin{align}
    \label{eq:gammat_ident}
    \tilde{\Gamma}^r_{k_1 k_2} \tilde{\Gamma}^r_{k_3 k_4} 
    - \tilde{\Gamma}^r_{k_1 k_3} \tilde{\Gamma}^r_{k_2 k_4} 
    + \tilde{\Gamma}^r_{k_1 k_4} \tilde{\Gamma}^r_{k_2 k_3} 
    = 0.
  \end{align}
\end{lemma}

\begin{proof}
  Define
  \begin{equation}
    \tilde{\mathcal{C}}
    = \begin{pmatrix}
      \tilde{A} \\ \tilde{B}
    \end{pmatrix}
    = \begin{pmatrix}
      L_{11} & L_{12}\\
      L_{21} & L_{22}
    \end{pmatrix}^\dag
    \begin{pmatrix}
      A \\ B
    \end{pmatrix}
    = \mathcal{L}^\dag \mathcal{C}.
  \end{equation}
  It is readily shown that
  \begin{equation}
    \tilde{\Gamma}^r_{k_1 k_2}
    = \tilde{B}^*_{r k_1} \tilde{A}^*_{r k_2} - \tilde{A}^*_{r k_1} \tilde{B}^*_{r k_2}.
  \end{equation}
  One can then verify Equation~\eqref{eq:gammat_ident} by straightforward algebraic manipulations.
\end{proof}

We now prove the robust Wick's theorem (Equation~(28)) starting from
the standard nonorthogonal Wick's theorem (Equation~(3)), assuming
$\Gamma^{(-)} = 0_{2d \times 2d}$.

\begin{proof}
  Let us first assume $\braket{\Phi_0|\Phi_1} \neq 0$, so the elements of $s$ defined
  in Equations~(21) and (22) are nonzero. By Equation~(27)
  and the definition of the Pfaffian (Equation~\eqref{eq:pfaffian}), we have
  \begin{align}
    \label{eq:pf_gamma1}
    \pf(\Gamma)
    = \frac{1}{2^d d!} \sum_{r_1 r_2 \cdots r_d} s_{r_1}^{-1} s_{r_2}^{-1} \cdots s_{r_d}^{-1}
    \sum_{\sigma \in S_{2d}} \text{sgn}(\sigma)
    \tilde{\Gamma}^{r_1}_{\sigma(1) \sigma(2)}
    \tilde{\Gamma}^{r_2}_{\sigma(3) \sigma(4)}
    \cdots
    \tilde{\Gamma}^{r_d}_{\sigma(2d-1) \sigma(2d)}.
  \end{align}

  By group properties, the mapping
  \begin{align*}
    f_1: S_{2d} &\to S_{2d}\\
    \sigma &\mapsto \sigma' = \sigma \cdot (23)
  \end{align*}
  is bijective, where $(23)$ denotes the transposition between $2$ and $3$.
  Similarly,
  \begin{align*}
    f_2: S_{2d} &\to S_{2d}\\
    \sigma &\mapsto \sigma'' = \sigma \cdot (24)
  \end{align*}
  is also bijective.
  
  For a given summant in the outer summation in Equation~\eqref{eq:pf_gamma1}, if $r_1 = r_2 = r$,
  we have
  \begin{subequations}
  \begin{align}
    \nonumber
    &\quad\sum_{\sigma \in S_{2d}} \text{sgn}(\sigma)
    \tilde{\Gamma}^{r}_{\sigma(1) \sigma(2)}
    \tilde{\Gamma}^{r}_{\sigma(3) \sigma(4)}
    \tilde{\Gamma}^{r_3}_{\sigma(5) \sigma(6)}
    \cdots
    \tilde{\Gamma}^{r_d}_{\sigma(2d-1) \sigma(2d)}\\
    \nonumber
    &= \frac{1}{3} \sum_{\sigma \in S_{2d}} 
    \left(
      \text{sgn}(\sigma)
      \tilde{\Gamma}^{r}_{\sigma(1) \sigma(2)}
      \tilde{\Gamma}^{r}_{\sigma(3) \sigma(4)}
      + \text{sgn}(\sigma')
      \tilde{\Gamma}^{r}_{\sigma'(1) \sigma'(2)}
      \tilde{\Gamma}^{r}_{\sigma'(3) \sigma'(4)}
      + \text{sgn}(\sigma'')
      \tilde{\Gamma}^{r}_{\sigma''(1) \sigma''(2)}
      \tilde{\Gamma}^{r}_{\sigma''(3) \sigma''(4)}
    \right)\\
    &\qquad\qquad\qquad \times
      \tilde{\Gamma}^{r_3}_{\sigma(5) \sigma(6)}
      \cdots
      \tilde{\Gamma}^{r_d}_{\sigma(2d-1) \sigma(2d)}\\
    \nonumber
    &= \frac{1}{3} \sum_{\sigma \in S_{2d}} \text{sgn}(\sigma)
    \left(
      \tilde{\Gamma}^{r}_{\sigma(1) \sigma(2)}
      \tilde{\Gamma}^{r}_{\sigma(3) \sigma(4)}
      -
      \tilde{\Gamma}^{r}_{\sigma(1) \sigma(3)}
      \tilde{\Gamma}^{r}_{\sigma(2) \sigma(4)}
      -
      \tilde{\Gamma}^{r}_{\sigma(1) \sigma(4)}
      \tilde{\Gamma}^{r}_{\sigma(3) \sigma(2)}
    \right)\\
    &\qquad\qquad\qquad \times
      \tilde{\Gamma}^{r_3}_{\sigma(5) \sigma(6)}
      \cdots
      \tilde{\Gamma}^{r_d}_{\sigma(2d-1) \sigma(2d)}\\
    \nonumber
    &= \frac{1}{3} \sum_{\sigma \in S_{2d}} \text{sgn}(\sigma)
    \left(
      \tilde{\Gamma}^{r}_{\sigma(1) \sigma(2)}
      \tilde{\Gamma}^{r}_{\sigma(3) \sigma(4)}
      -
      \tilde{\Gamma}^{r}_{\sigma(1) \sigma(3)}
      \tilde{\Gamma}^{r}_{\sigma(2) \sigma(4)}
      +
      \tilde{\Gamma}^{r}_{\sigma(1) \sigma(4)}
      \tilde{\Gamma}^{r}_{\sigma(2) \sigma(3)}
    \right)\\
    &\qquad\qquad\qquad \times
      \tilde{\Gamma}^{r_3}_{\sigma(5) \sigma(6)}
      \cdots
      \tilde{\Gamma}^{r_d}_{\sigma(2d-1) \sigma(2d)}\\
    &= 0,
  \end{align}
  \end{subequations}
  where we have used Lemma~\ref{lemma:cancellation}.

  Therefore, this outer summation can be restricted to distinct $r_1, r_2, \cdots, r_d$ without affecting
  the result. It follows from the nonorthogonal Wick's theorem and the definition of the hyper-Pfaffian
  (Equation~\eqref{eq:hyperpfaffian}, see also Equation~\eqref{eq:hyperpfaffian_special}) that
  \begin{subequations}
  \begin{align}
    \nonumber
    &\quad\:\braket{\Phi_0|\gamma^1 \gamma^2 \cdots \gamma^{2d}|\Phi_1}\\
    &= \braket{\Phi_0|\Phi_1} \pf(\Gamma)\\
    &= \frac{1}{2^d}\: \zeta \prod_r s_r 
    \sum_{r_1 < r_2 < \cdots < r_d} s_{r_1}^{-1} s_{r_2}^{-1} \cdots s_{r_d}^{-1}
    \sum_{\sigma \in S_{2d}} \text{sgn}(\sigma)\,
    \tilde{\Gamma}^{r_1}_{\sigma(1) \sigma(2)}
    \tilde{\Gamma}^{r_2}_{\sigma(3) \sigma(4)}
    \cdots
    \tilde{\Gamma}^{r_d}_{\sigma(2d-1) \sigma(2d)}\\
    &= \frac{(2d)!}{2^d}\: \zeta \sum_{r_1 < r_2 < \cdots < r_d} \lambda_{r_1 r_2 \cdots r_d}\:
    \hpf(\tilde{\Gamma}^{r_1} \otimes \tilde{\Gamma}^{r_2} \otimes \cdots \otimes \tilde{\Gamma}^{r_d})\\
    &= \frac{(2d)!}{2^d d!}\: \zeta \sum_{r_1 r_2 \cdots r_d} \lambda_{r_1 r_2 \cdots r_d}\:
    \hpf(\tilde{\Gamma}^{r_1} \otimes \tilde{\Gamma}^{r_2} \otimes \cdots \otimes \tilde{\Gamma}^{r_d})\\
    &= (2d-1)!!\: \zeta \sum_{r_1 r_2 \cdots r_d} \lambda_{r_1 r_2 \cdots r_d}\:
    \hpf(\tilde{\Gamma}^{r_1} \otimes \tilde{\Gamma}^{r_2} \otimes \cdots \otimes \tilde{\Gamma}^{r_d})
    \label{eq:robust_wick2}
  \end{align}
  \end{subequations}
\end{proof}

All of the $s_r^{-1}$ factors get cancelled out in this final expression; therefore, as
$\braket{\Phi_0|\Phi_1} \to 0$, or equivalently $s_r \to 0$ for some $1 \leq r \leq M$,
Equation~\eqref{eq:robust_wick2} still yields a well-defined value of the matrix element.
Namely, the robust Wick's theorem is valid regardless of whether $\ket{\Phi_0}$ and $\ket{\Phi_1}$
are orthogonal or not.

As noted in the main text, Lemma~\ref{lemma:cancellation} and the properties of the
hyper-Pfaffian guarantee the cancellation of self-interaction. The definition of
$\lambda_{r_1 r_2 \cdots r_d}$ in Equation~(29) further enforces this cancellation numerically.

\section{Hartree--Fock Limit}

Following our robust normalization procedure in Section~V,
Equation~(42) reduces to
\begin{align}
  \bar{U}^{j\prime}
  = 0_{M_O \times M_O},
  \quad
  \bar{V}^{j\prime}
  = I_{M_O}
\end{align}
in the Hartree--Fock (HF) limit,
where $M_O = M^0_C = M^1_C$ is the number of occupied spin-orbitals in the
Slater determinant $\ket{\Phi'_0}$ or $\ket{\Phi'_1}$.
It follows from Equation~(41) that
\begin{equation}
  U^{j\prime} = 0_{M \times M_O}, \quad
  V^{j\prime} = \left(D^{j\prime}\right)^*,
\end{equation}
where columns of $D^{j\prime}$ are the coefficients of the occupied canonical orbitals.
According to Equation~(44), we have
\begin{equation}
  \mathcal{S}'
  = \begin{pmatrix}
    0_{M_O \times M_O} &
    \left(D^{0\prime}\right)^\dag D^{1\prime}\\
    -\left(D^{1\prime}\right)^\T \left(D^{0\prime}\right)^* &
    0_{M_O \times M_O}
  \end{pmatrix}.
\end{equation}
Note that $\left(D^{0\prime}\right)^\dag D^{1\prime}$ is the orbital overlap matrix,
whose singular value decomposition (SVD) is given by
\begin{equation}
  \left(D^{0\prime}\right)^\dag D^{1\prime}
  = X \, \text{diag}(s') \, Y^\dag,
\end{equation}
where elements of $s'$ are the singular values, and $X$ and $Y$ are the
left- and right-singular vectors. It is readily seen that
\begin{equation}
  \mathcal{S}'
  = \mathcal{Q}' \bar{\mathcal{S}}' \mathcal{Q}^{\prime\T}
\end{equation}
with
\begin{equation}
  \bar{\mathcal{S}}'
  = \begin{pmatrix}
    0_{M_O \times M_O} & \text{diag}(s')\\
    -\text{diag}(s') & 0_{M_O \times M_O}
  \end{pmatrix}, \quad
  \mathcal{Q}'
  = \begin{pmatrix}
    0_{M_O \times M_O} & X\\
    -Y^* & 0_{M_O \times M_O}
  \end{pmatrix}.
\end{equation}
Namely, the canonical values of $\mathcal{S}'$ reduces to the singular values of the orbital overlap matrix
in the HF limit.

Moreover, in this limit, Equation~(26) reduces to
\begin{equation}
  \mathcal{L}'
  = \begin{pmatrix}
    L'_{11} & L'_{12}\\
    L'_{21} & L'_{22}
  \end{pmatrix}
  = \mathcal{V}' \mathcal{Q}^{\prime *},
\end{equation}
where
\begin{equation}
  L'_{12} = L'_{21} = 0_{M_O \times M_O}, \quad
  L'_{11} = D^{1\prime} Y, \quad 
  L'_{22} = \left(D^{0\prime} X\right)^*.
\end{equation}
One can then easily verify that the $1$- and $2$-RTM expressions in \cite{Koch1993LinearSuperposition}
are reproduced by the robust Wick's theorem.

\bibliographystyle{nsf}
\bibliography{ref}